\documentclass[bjps]{imsart}

\RequirePackage{amsthm,amsmath,amsfonts,amssymb}
\RequirePackage[authoryear]{natbib}
\RequirePackage[colorlinks,citecolor=blue,urlcolor=blue]{hyperref}
\RequirePackage{booktabs,array,tabularx,enumerate,adjustbox}
\RequirePackage{xcolor}

\startlocaldefs
\numberwithin{equation}{section}
\theoremstyle{plain}
\newtheorem{theorem}{Theorem}
\newtheorem{corollary}{Corollary}
\newtheorem{proposition}{Proposition}

\theoremstyle{definition}
\newtheorem{assumption}{Assumption}
\newtheorem{definition}{Definition}
\newtheorem{remark}{Remark}

\newcommand{\E}{\mathbb{E}}
\newcommand{\Prob}{\mathbb{P}}
\newcommand{\R}{\mathbb{R}}
\newcommand{\calD}{\mathcal{D}}
\newcommand{\calX}{\mathcal{X}}
\newcommand{\calI}{\mathcal{I}}
\newcommand{\ind}{\mathbf{1}}
\newcommand{\TV}{\mathrm{TV}}
\newcommand{\hC}{\widehat{C}}
\newcommand{\hQ}{\widehat{Q}}
\newcommand{\hmu}{\widehat{\mu}}
\newcommand{\htau}{\widehat{\tau}}
\newcommand{\he}{\widehat{e}}
\newcommand{\hG}{\widehat{\Gamma}}
\newcommand{\runin}[1]{\par\medskip\noindent\textbf{#1}\ }
\newcommand{\tabnote}[1]{\par\vspace{2pt}\begin{minipage}{0.97\textwidth}\scriptsize #1\end{minipage}}
\endlocaldefs

\begin{document}
	
	\begin{frontmatter}
		\title{Localized conformal inference for individual treatment effects}
		\runtitle{Localized conformal inference for ITEs}
		
		\begin{aug}
			\author[A]{\inits{F.}\fnms{Fernando}~\snm{Delbianco}\ead[label=e1]{fernando.delbianco@uns.edu.ar}}
			\author[A]{\inits{F.}\fnms{Fernando}~\snm{Tohm\'e}\ead[label=e2]{ftohme@criba.edu.ar}}
			\address[A]{Department of Economics, Universidad Nacional del Sur, and Institute of Mathematics (INMABB), CONICET, Bah\'ia Blanca, Argentina\printead[presep={,\ }]{e1,e2}}
		\end{aug}
		
		\begin{abstract}
			Conformal prediction intervals for individual treatment effects (ITEs) usually guarantee marginal coverage, averaged over the covariate distribution, and can undercover badly where effects are most heterogeneous: in our simulations, a globally calibrated conformal meta-learner reaches 94.5\% marginal coverage at the 90\% level but only 62.3\% in the worst covariate slab. We propose a localized split-conformal procedure that calibrates counterfactual intervals with kernel weights centered at the covariate value of the query, using a kernel built from causal-forest variable importance. The intervals have exact finite-sample coverage for units drawn from a kernel neighborhood of the query, need no propensity model, and cover the query itself with probability at least the nominal level minus an explicit localization error, which makes the bandwidth trade-off interpretable. We give a bandwidth-selection rule that uses only the training fold and a variant that calibrates doubly robust pseudo-outcomes. In simulations with 20 covariates, localization raises the worst-slab coverage of out-of-sample ITE intervals from 0.58--0.87 for global methods to 0.98 at a similar length, with 4.7\% of intervals infinite. On the IHDP benchmark, where local information is scarce, the procedure reduces to global calibration.
		\end{abstract}
		
		\begin{keyword}
			\kwd{Conformal prediction}
			\kwd{individual treatment effect}
			\kwd{localized calibration}
			\kwd{covariate shift}
			\kwd{causal forest}
			\kwd{treatment effect heterogeneity}
			\kwd{doubly robust score}
		\end{keyword}
		
	\end{frontmatter}
	
	\section{Introduction}
	\label{sec:intro}
	
	Most of the heterogeneous-treatment-effect literature targets the conditional average treatment effect (CATE), $\tau(x)=\E[Y(1)-Y(0)\mid X=x]$. Modern estimators such as causal forests \citep{wager2018,athey2019grf}, the X-learner \citep{kunzel2019}, the R-learner \citep{nie2021}, and the DR-learner \citep{kennedy2023} deliver accurate point estimates and, under regularity conditions, asymptotically valid confidence intervals for $\tau(x)$. These intervals do not answer the question a decision maker often asks about a specific unit $q$: what range of values is plausible for its own effect $\tau_q=Y_q(1)-Y_q(0)$? When $\mathrm{Var}(Y(1)-Y(0)\mid X=x)$ is large, a confidence interval for $\tau(x)$ can be narrow and yet contain $\tau_q$ with probability far below any nominal level.
	
	Conformal prediction \citep{vovk2005,lei2018} provides finite-sample prediction intervals without distributional assumptions beyond exchangeability, and \citet{lei2021} showed how to adapt it to counterfactuals and ITEs by reweighting a calibration set with inverse propensity scores. Their guarantee is marginal: coverage is averaged over the covariate distribution. Exact finite-sample coverage conditional on $X=x_q$ is impossible for any nontrivial procedure \citep{vovk2012,barber2021limits}, so marginal coverage is the natural default. The difficulty is that marginal coverage can conceal large local failures. In the simulation study of Section~\ref{sec:sim}, the conformal doubly robust meta-learner of \citet{alaa2023} has marginal coverage $0.945$ at nominal level $0.90$ but covers only $62.3\%$ of the ITEs in the worst decile slab of the covariates, and the weighted conformal intervals of \citet{lei2021} cover $87.3\%$ marginally and $57.9\%$ in the worst slab. For in-sample queries, unweighted calibration covers $85.3\%$ marginally and $35.3\%$ in the worst slab. The units that are undercovered are precisely those whose effects are most heterogeneous, which are often the units for which an individual interval matters most. This paper asks how much of this gap a conformal procedure can close, and at what cost.
	
	Our approach calibrates counterfactual intervals with kernel weights centered at $x_q$. A kernel neighborhood of $x_q$ is a covariate shift of the calibration distribution, so weighted conformal inference with fixed weights \citep{tibshirani2019} delivers exact finite-sample coverage for a unit drawn from that neighborhood. This step is a direct application of existing results, and we do not claim it as a contribution. The contribution lies between the neighborhood and the point $x_q$. We define a \emph{localization error}, the kernel-averaged Kolmogorov distance between the conformity-score distribution at a neighbor and at $x_q$, and show that coverage at $x_q$ itself is at least $1-\alpha$ minus this error. The error vanishes when the score distribution does not vary across the neighborhood and is bounded by a Lipschitz constant times the bandwidth under smoothness. This makes the bandwidth trade-off explicit and checkable: a narrow kernel reduces the localization error, while the weighted quantile needs a minimum kernel mass of calibration units to be finite and an adequate effective sample size to be stable. The same quantities are the basis of a bandwidth-selection rule.
	
	The contributions are the following.
	\begin{enumerate}[(1)]
		\item A pointwise coverage bound for localized counterfactual and ITE intervals in terms of the localization error (Theorem~\ref{thm:local}(ii), Proposition~\ref{prop:lipschitz}, and Corollary~\ref{cor:ite}), for in-sample queries, where one potential outcome is observed, and for out-of-sample queries, where both are missing. Calibration is arm specific and requires no propensity model.
		\item A practical implementation: an importance-weighted kernel built from a causal forest, so that the neighborhood is defined by the covariates that drive effect heterogeneity, and a bandwidth-selection rule that uses only the training fold and enforces the existence of a finite interval and a minimum effective sample size (Section~\ref{sec:impl}).
		\item A localized version of conformal meta-learners \citep{alaa2023}, which calibrates doubly robust pseudo-outcome scores and covers the ITE under a local stochastic-dominance assumption (Theorem~\ref{thm:pseudo}).
		\item Evidence on local coverage: a simulation study with unit-level effect heterogeneity that reports worst-slab coverage of ten methods, and an evaluation on the IHDP benchmark (Sections~\ref{sec:sim} and \ref{sec:ihdp}), with code that reproduces every table.
	\end{enumerate}
	Table~\ref{tab:contrib} separates the components of the procedure that are taken from the literature from those that are new.
	
	\begin{table}[t]
		\footnotesize
		\centering
		\caption{Components of the proposed procedure: what is taken from the literature and what is new.}\label{tab:contrib}
		\renewcommand{\arraystretch}{1.15}
		\begin{tabularx}{\textwidth}{>{\raggedright\arraybackslash}X >{\raggedright\arraybackslash}p{3.4cm} >{\raggedright\arraybackslash}X}
			\toprule
			Component & Source & Status in this paper \\
			\midrule
			Weighted split conformal inference with fixed weights under covariate shift & \citet{tibshirani2019} & Used as is; Theorem~\ref{thm:local}(i) applies it to a kernel-defined target population \\
			Counterfactual intervals; ITE intervals by subtraction (in sample) and Bonferroni combination (out of sample) & \citet{lei2021} & Used as is (Corollary~\ref{cor:ite}), with localization weights in place of propensity weights \\
			Kernel-localized calibration and localized coverage & \citet{guan2023}; \citet{hore2025} & Adapted: kernel center fixed at $x_q$ and kernel fixed on the training fold; arm-specific calibration sets; no propensity model \\
			Localization error and pointwise coverage bound & --- & New: Theorem~\ref{thm:local}(ii), Proposition~\ref{prop:lipschitz}, Corollary~\ref{cor:ite}, Lemma~S1 \\
			Existence condition and bandwidth selection on the training fold & --- & New: Remark~\ref{rem:tradeoff} and Section~\ref{sec:bandwidth} \\
			Importance-weighted Gaussian kernel and forest kernel & Forest weights of \citet{athey2019grf} & New use as localization kernels (Section~\ref{sec:kernels}) \\
			Pseudo-outcome calibration under stochastic dominance & \citet{alaa2023} & Adapted: localized version with a pointwise bound (Theorem~\ref{thm:pseudo}) \\
			Local (worst-slab) coverage of ITE intervals across methods & --- & New evidence (Sections~\ref{sec:sim} and \ref{sec:ihdp}) \\
			\bottomrule
		\end{tabularx}
	\end{table}
	
	The paper is organized as follows. Section~\ref{sec:related} positions the paper in the literature. Section~\ref{sec:setup} defines the targets. Section~\ref{sec:method} presents the localized procedure and its guarantees. Section~\ref{sec:pseudo} analyzes the pseudo-outcome variant. Section~\ref{sec:impl} describes the implementation and the choice of tuning parameters. Sections~\ref{sec:sim} and \ref{sec:ihdp} report the simulation study and the IHDP benchmark. Section~\ref{sec:discussion} concludes. Proofs are in the Appendix.
	
	\section{Related work}
	\label{sec:related}
	
	\subsection{Conformal inference for counterfactuals and ITEs}
	\citet{lei2021} proposed weighted split conformal inference for counterfactuals. Treated units form the calibration set for $Y(1)$, and weights $w(x)=1/e(x)$ (or $w(x)=(1-e(x))/e(x)$) correct the covariate shift between the treated arm and the target population \citep{tibshirani2019}. Coverage is exact when $e$ is known and approximately valid when $e$ is estimated, with an error bounded by the $L_1$ distance between estimated and true weights. For ITEs of in-sample units, they subtract the observed outcome from the counterfactual interval. For new units, they combine two counterfactual intervals, either by a Bonferroni split or by a nested procedure that trades some conservativeness for a single-stage construction. Their guarantee is marginal over $P_X$. \citet{kivaranovic2020} studied the randomized setting, and \citet{jin2023} and \citet{yin2024} extended conformal ITE inference to sensitivity analysis under unmeasured confounding. \citet{alaa2023} introduced conformal meta-learners, which calibrate residuals of pseudo-outcomes rather than potential outcomes, and showed that the resulting intervals cover the ITE when the pseudo-outcome conformity score stochastically dominates the oracle ITE score. They prove a distribution-free convex-order dominance for the doubly robust pseudo-outcome that implies validity at the coverage levels used in practice. Our paper differs from all of these in its target: coverage at or near a fixed $x_q$ rather than averaged over $P_X$.
	
	\subsection{Localized and conditional conformal prediction}
	Exact distribution-free conditional coverage is impossible \citep{vovk2012,barber2021limits}, so a substantial literature studies relaxations. \citet{leiwasserman2014} used kernel-localized conformal bands with asymptotic conditional validity. \citet{guan2023} introduced localized conformal prediction (LCP), which reweights calibration scores by a kernel centered at the test point and adds a calibration step to recover exact marginal coverage despite the test-point-dependent weights. \citet{hore2025} showed that randomizing the kernel center yields a ``relaxed local'' guarantee: exact coverage conditional on a randomized location near the test point, as well as validity under covariate shift. \citet{gibbs2025} obtained exact coverage simultaneously over a class of covariate shifts, and \citet{romano2020malice} studied group-conditional coverage. \citet{barber2023beyond} bound the coverage loss of weighted conformal inference with fixed weights by a weighted total variation distance.
	
	Our localized guarantee is of the same type as those of \citet{guan2023} and \citet{hore2025}: coverage for a unit drawn from a kernel neighborhood of $x_q$. The differences are that (a) we work with counterfactuals, so the calibration sets are arm specific and the target involves unobserved potential outcomes; (b) we avoid the test-point-dependent-weights problem by fixing the kernel center at $x_q$ and the kernel itself on the training fold, so the weights are a fixed function of $x$ as required by \citet{tibshirani2019}; and (c) we give an explicit pointwise bound at $x_q$ in terms of a localization error, which makes the bandwidth trade-off interpretable.
	
	\subsection{CATE estimation and doubly robust scores}
	Causal forests \citep{wager2018,athey2019grf} provide both the point estimate $\htau(x)$ and, through their adaptive neighborhood weights, a natural localization kernel. The augmented inverse probability weighted (AIPW) score \citep{robins1994} is doubly robust and Neyman orthogonal \citep{chernozhukov2018}, and it underlies the DR-learner \citep{kennedy2023}. We use these tools for point prediction and for constructing localization weights. We stress that orthogonality is a property of point estimation. It does not by itself deliver coverage of the ITE, which is why the pseudo-outcome variant in Section~\ref{sec:pseudo} requires a separate assumption.
	
	\section{Setup and targets}
	\label{sec:setup}
	
	\subsection{Potential outcomes and the query}
	
	Let $(X_i,T_i,Y_i(0),Y_i(1))$, $i=1,\dots,n$, be independent and identically distributed draws from a distribution $P$, with $X_i\in\calX\subset\R^p$, $T_i\in\{0,1\}$, and $Y_i(t)\in\R$. We observe $\calD=\{(X_i,T_i,Y_i)\}_{i=1}^n$ with $Y_i=Y_i(T_i)$. Write $e(x)=\Prob(T=1\mid X=x)$ for the propensity score, $P_X$ for the marginal law of $X$, and $P_{X\mid t}$ for the conditional law of $X$ given $T=t$.
	
	\begin{assumption}[Consistency and no interference]
		\label{ass:sutva}
		$Y_i=Y_i(T_i)$, and the potential outcomes of unit $i$ do not depend on the treatments of other units.
	\end{assumption}
	
	\begin{assumption}[Unconfoundedness and overlap]
		\label{ass:ignorability}
		$(Y_i(0),Y_i(1))\perp T_i\mid X_i$, and there is $\eta\in(0,1/2)$ with $\eta\le e(x)\le 1-\eta$ for $P_X$-almost every $x$.
	\end{assumption}
	
	The object of interest is the individual treatment effect of a query unit $q$ with covariate value $x_q$,
	\[
	\tau_q = Y_q(1)-Y_q(0).
	\]
	Two cases arise. An \emph{in-sample} query is a unit of the observed data set, so $T_q$ and $Y_q=Y_q(T_q)$ are observed and only $Y_q(1-T_q)$ is missing. An \emph{out-of-sample} query is a new unit with covariates $x_q$ and no observed outcome, so both potential outcomes are missing. The two cases require different constructions, and the second is strictly harder.
	
	\subsection{What coverage can mean for a query}
	
	For a set-valued map $\hC$ built from $\calD$, the marginal guarantee $\Prob(\tau_q\in\hC(X_q))\ge 1-\alpha$, with $X_q\sim P_X$, is the one studied by \citet{lei2021}. The pointwise statement
	\begin{equation}
		\label{eq:pointwise}
		\Prob\bigl(\tau_q\in\hC(x_q)\mid X_q=x_q\bigr)\ge 1-\alpha\quad\text{for all }x_q
	\end{equation}
	cannot be achieved by any procedure with finite expected length under continuous covariate distributions \citep{barber2021limits}. We therefore study two intermediate targets.
	
	\begin{definition}[Localized coverage]
		\label{def:localized}
		Let $K:\calX\times\calX\to[0,\infty)$ be a kernel with $K(x_q,x_q)\ge K(x,x_q)$ for all $x$, and let $P_h^{(t)}$ be the probability measure on $\calX$ with density proportional to $K(x,x_q)$ with respect to $P_{X\mid t}$. A counterfactual interval $\hC_t$ has localized coverage at level $1-\alpha$ around $x_q$ if, for an independent pair $(X',Y'(t))$ with $X'\sim P_h^{(t)}$ and $Y'(t)\mid X'\sim P_{Y(t)\mid X}$,
		\[
		\Prob\bigl(Y'(t)\in\hC_t(X')\bigr)\ge 1-\alpha .
		\]
	\end{definition}
	
	Localized coverage is the natural finite-sample relaxation of (\ref{eq:pointwise}): it is coverage averaged over ``units like $q$,'' with the kernel making ``like'' precise. The subscript $h$ records a bandwidth; as $h\to 0$, $P_h^{(t)}$ concentrates at $x_q$. The second target is the pointwise probability in (\ref{eq:pointwise}) itself, for which we give a lower bound of the form $1-\alpha-\varepsilon$ with an explicit $\varepsilon$.
	
	\section{Localized conformal inference for counterfactuals}
	\label{sec:method}
	
	\subsection{Sample splitting}
	
	Randomly split the index set into a training fold $\calD_{\mathrm{tr}}$ and a calibration fold $\calD_{\mathrm{cal}}$, independent of the data. Everything that is \emph{learned} is learned on $\calD_{\mathrm{tr}}$: outcome regressions $\hmu_t$, a propensity model $\he$, a CATE estimate $\htau$, variable-importance weights, and the localization kernel $K$ together with its bandwidth. All statements below are conditional on $\calD_{\mathrm{tr}}$, so $\hmu_t$, $\he$, $\htau$, and $K$ are treated as fixed functions.
	
	\subsection{The procedure}
	\label{sec:procedure}
	
	Fix an arm $t\in\{0,1\}$ and a query covariate value $x_q$. Let $\calI_t=\{i\in\calD_{\mathrm{cal}}:T_i=t\}$ and $n_t=|\calI_t|$.
	
	\begin{enumerate}[(i)]
		\item \emph{Scores.} For $i\in\calI_t$, compute $V_i=|Y_i-\hmu_t(X_i)|$. (Any score of the form $V(x,y)$ with $V$ fixed on $\calD_{\mathrm{tr}}$ can be used; the conformalized quantile regression score of \citet{romano2019cqr} is a natural alternative under heteroscedasticity.)
		\item \emph{Localization weights.} For $i\in\calI_t$, set $w_i=K(X_i,x_q)$, and write $\bar w=K(x_q,x_q)$ and $W_t=\sum_{i\in\calI_t}w_i$. For a generic covariate value $x$ define the normalized weights
		\[
		p_i(x)=\frac{w_i}{W_t+K(x,x_q)},\qquad p_\infty(x)=\frac{K(x,x_q)}{W_t+K(x,x_q)} .
		\]
		\item \emph{Weighted quantile.} Let
		\[
		\hQ_t(x)=\inf\Bigl\{s\in\R:\ \sum_{i\in\calI_t}p_i(x)\,\ind\{V_i\le s\}\ \ge\ 1-\alpha\Bigr\},
		\]
		with $\hQ_t(x)=+\infty$ if the set is empty, which happens exactly when $p_\infty(x)>\alpha$. This is the $(1-\alpha)$-quantile of the discrete measure $\sum_i p_i(x)\delta_{V_i}+p_\infty(x)\delta_{+\infty}$; with constant weights it reduces to the $\lceil (n_t+1)(1-\alpha)\rceil$-th order statistic of the scores, which is the standard finite-sample correction of split conformal inference.
		\item \emph{Counterfactual interval.} $\hC_t(x)=\bigl[\hmu_t(x)-\hQ_t(x),\ \hmu_t(x)+\hQ_t(x)\bigr]$.
	\end{enumerate}
	The interval reported for the query is $\hC_t(x_q)$, for which $K(x_q,x_q)=\bar w$ and $p_\infty(x_q)=\bar w/(W_t+\bar w)$. Because $K$ and $\hmu_t$ are fixed functions given $\calD_{\mathrm{tr}}$, the weights are a fixed function of $X_i$, as required by weighted exchangeability.
	
	\subsection{Localization error and guarantees}
	
	Let $G_{t,x}(s)=\Prob\bigl(|Y(t)-\hmu_t(x)|\le s\ \big|\ X=x\bigr)$ denote the conditional distribution function of the score at covariate value $x$ (conditional on $\calD_{\mathrm{tr}}$). Define the \emph{localization error}
	\begin{equation}
		\label{eq:eps}
		\varepsilon_t(h;x_q)=\int \sup_{s\ge 0}\bigl|G_{t,x}(s)-G_{t,x_q}(s)\bigr|\,dP_h^{(t)}(x).
	\end{equation}
	This is the kernel-averaged Kolmogorov distance between the score distribution at a neighbor $x$ and at $x_q$. It is zero if the score distribution does not depend on $x$ within the neighborhood, and it is small when that distribution is smooth in $x$ and the kernel is narrow.
	
	\begin{theorem}[Localized and pointwise coverage of counterfactual intervals]
		\label{thm:local}
		Suppose Assumptions~\ref{ass:sutva} and \ref{ass:ignorability} hold, and condition on $\calD_{\mathrm{tr}}$. Fix $t\in\{0,1\}$, $x_q\in\calX$, and $\alpha\in(0,1)$, and construct $\hC_t$ as in Section~\ref{sec:procedure}.
		\begin{enumerate}[(i)]
			\item \textup{(Localized coverage.)} For an independent pair $(X',Y'(t))$ with $X'\sim P_h^{(t)}$ and $Y'(t)\mid X'\sim P_{Y(t)\mid X}$,
			\[
			\Prob\bigl(Y'(t)\in\hC_t(X')\bigr)\ \ge\ 1-\alpha .
			\]
			If in addition the scores $V_i$ are almost surely distinct, then $\Prob(Y'(t)\in\hC_t(X'))\le 1-\alpha+\E\bigl[\max_{i\in\calI_t}p_i(X')\vee p_\infty(X')\bigr]$.
			\item \textup{(Pointwise coverage.)} Let $q$ be either an out-of-sample unit with $X_q=x_q$, or a unit with $X_q=x_q$ and $T_q=1-t$ that belongs to $\calD_{\mathrm{cal}}$ or is held out from both folds. Then
			\[
			\Prob\bigl(Y_q(t)\in\hC_t(x_q)\ \big|\ X_q=x_q\ [,\,T_q=1-t]\bigr)\ \ge\ 1-\alpha-\varepsilon_t(h;x_q).
			\]
		\end{enumerate}
	\end{theorem}
	
	Part (i) is exact and distribution free. It does not involve the propensity score: the calibration units in $\calI_t$ are draws from $P_{X\mid t}\otimes P_{Y(t)\mid X}$ (the joint distribution of a covariate and a potential outcome), the localized target is a covariate shift of $P_{X\mid t}$ with likelihood ratio proportional to $K(\cdot,x_q)$, and unconfoundedness is what makes the conditional law of $Y(t)$ given $X$ the same in the calibration arm and at the target. Part (ii) converts the localized statement into a statement about the query itself. The only cost is $\varepsilon_t$, and the only place where the arm-specific reference measure enters is through $\varepsilon_t$, which is an average over the neighborhood.
	
	\begin{corollary}[ITE intervals]
		\label{cor:ite}
		Under the conditions of Theorem~\ref{thm:local}:
		\begin{enumerate}[(a)]
			\item \textup{(In-sample query.)} If $q\in\calD_{\mathrm{cal}}$ with $T_q=0$, let $\hC(\tau_q)=\hC_1(x_q)-Y_q$; if $T_q=1$, let $\hC(\tau_q)=Y_q-\hC_0(x_q)$. Then $\Prob(\tau_q\in\hC(\tau_q)\mid X_q=x_q,T_q)\ge 1-\alpha-\varepsilon_{1-T_q}(h;x_q)$.
			\item \textup{(Out-of-sample query.)} Construct $\hC_1(x_q)=[L_1,U_1]$ at level $\alpha_1$ and $\hC_0(x_q)=[L_0,U_0]$ at level $\alpha_0$ with $\alpha_1+\alpha_0=\alpha$, and let $\hC(\tau_q)=[L_1-U_0,\ U_1-L_0]$. Then $\Prob(\tau_q\in\hC(\tau_q)\mid X_q=x_q)\ge 1-\alpha-\varepsilon_1(h;x_q)-\varepsilon_0(h;x_q)$.
		\end{enumerate}
	\end{corollary}
	
	\subsection{Controlling the localization error}
	
	\begin{proposition}[Lipschitz bound]
		\label{prop:lipschitz}
		Let $L_{t,x}$ denote the conditional law of the residual $Y(t)-\hmu_t(x)$ given $X=x$. Then $\sup_s|G_{t,x}(s)-G_{t,x_q}(s)|\le\TV(L_{t,x},L_{t,x_q})$, so
		\[
		\varepsilon_t(h;x_q)\le\int\TV(L_{t,x},L_{t,x_q})\,dP_h^{(t)}(x).
		\]
		If, for a metric $d$ on $\calX$ and a constant $\lambda_t$, $\TV(L_{t,x},L_{t,x_q})\le\lambda_t\,d(x,x_q)$ for all $x$ in the support of $K(\cdot,x_q)$, then $\varepsilon_t(h;x_q)\le\lambda_t\int d(x,x_q)\,dP_h^{(t)}(x)$, which is at most $\lambda_t h$ when $K(\cdot,x_q)$ is supported in the ball of $d$-radius $h$ around $x_q$.
	\end{proposition}
	
	The Lipschitz condition holds, for example, in a location-scale model $Y(t)=m_t(X)+\sigma_t(X)\xi_t$ with $\xi_t\perp X$ having an absolutely continuous density $f$ with $\int|f'(u)|\,du<\infty$ and $\int|uf'(u)|\,du<\infty$, whenever $\sigma_t$ is bounded away from zero and $m_t-\hmu_t$ and $\sigma_t$ are Lipschitz in $x$; $\lambda_t$ then scales with the Lipschitz constants of the prediction error and of the conditional scale. Lemma~S1 of the Supplementary Material \citep{supp} gives the constant $\lambda_t$ explicitly.
	
	If the bandwidth $h=h_n$ tends to zero and the Lipschitz constant $\lambda_t$ of Proposition~\ref{prop:lipschitz} is bounded uniformly in $n$ (recall that $\hmu_t$ changes with $n$), Proposition~\ref{prop:lipschitz} and Theorem~\ref{thm:local}(ii) give $\liminf_n\Prob(Y_q(t)\in\hC_t(x_q)\mid X_q=x_q)\ge1-\alpha$. If, in addition, $h_n$ shrinks slowly enough that $\E[\max_i p_i(X')\vee p_\infty(X')]\to0$, the upper bound in Theorem~\ref{thm:local}(i) shows that the localized coverage converges to $1-\alpha$.
	
	\begin{remark}[Bias--variance trade-off and existence of a finite interval]
		\label{rem:tradeoff}
		Theorem~\ref{thm:local}(ii) and Proposition~\ref{prop:lipschitz} quantify the bias side of localization: $\varepsilon_t$ decreases with $h$. The variance side appears through the weights. The interval is finite only if $p_\infty(x_q)\le\alpha$, that is,
		\[
		\sum_{i\in\calI_t}K(X_i,x_q)\ \ge\ \frac{1-\alpha}{\alpha}\,K(x_q,x_q),
		\]
		so with a kernel bounded by one at its center, the kernel-weighted count of arm-$t$ calibration units must be at least $(1-\alpha)/\alpha$, which is $9$ at $\alpha=0.1$. Beyond existence, the sampling variability of $\hQ_t$ is governed by the effective sample size $n_{\mathrm{eff},t}=W_t^2/\sum_i w_i^2$, and the upper bound in Theorem~\ref{thm:local}(i) shows that the interval becomes conservative when the largest normalized weight is not small.
	\end{remark}
	
	\begin{remark}[Targeting a localized version of $P_X$]
		\label{rem:propensity}
		If one prefers the reference measure to be $P_X$ rather than $P_{X\mid t}$, replace $w_i$ by $K(X_i,x_q)/e(X_i)$ for $t=1$ and by $K(X_i,x_q)/(1-e(X_i))$ for $t=0$. With the true $e$, Theorem~\ref{thm:local}(i) holds with $P_h^{(t)}$ replaced by the measure with density proportional to $K(\cdot,x_q)$ with respect to $P_X$. With an estimate $\he$ fitted on $\calD_{\mathrm{tr}}$, coverage degrades by at most one half of the expected $L_1$ distance between the normalized estimated and true weights, as in Theorem~3 of \citet{lei2021}.
	\end{remark}
	
	\begin{remark}[Relation to LCP and RLCP]
		\label{rem:lcp}
		In \citet{guan2023}, the kernel is centered at the test point, so the weights depend on the test covariates and an extra calibration step is needed. In our setting the query is a fixed $x_q$ and the kernel is fixed on the training fold, so the weights are a fixed function of the calibration covariates and the covariate-shift argument of \citet{tibshirani2019} applies directly. The randomized construction of \citet{hore2025} (RLCP) delivers exact coverage conditional on a randomized location drawn from the kernel around the test point. Theorem~\ref{thm:local}(i) is the deterministic analog for the localized population, and part (ii) adds the pointwise bound.
	\end{remark}
	
	\section{Calibrating doubly robust pseudo-outcome scores}
	\label{sec:pseudo}
	
	The procedure of Section~\ref{sec:method} calibrates residuals of potential outcomes. An alternative, which uses both arms in a single calibration set and is natural when the point estimate is a CATE learner, is to calibrate residuals of pseudo-outcomes. With nuisance functions $\hmu_0,\hmu_1,\he$ and a CATE estimate $\htau$ fitted on $\calD_{\mathrm{tr}}$, define for $i\in\calD_{\mathrm{cal}}$ the AIPW pseudo-outcome
	\begin{equation}
		\label{eq:aipw}
		\hG_i=\htau(X_i)+\frac{T_i-\he(X_i)}{\he(X_i)\{1-\he(X_i)\}}\bigl\{Y_i-\hmu_{T_i}(X_i)\bigr\},
	\end{equation}
	which is the form returned by the \texttt{grf} package; when $\htau=\hmu_1-\hmu_0$ it coincides with the DR-learner pseudo-outcome of \citet{kennedy2023}. Its score is $S_i=|\hG_i-\htau(X_i)|$. The localized procedure of Section~\ref{sec:procedure} is applied with $\calI_t$ replaced by $\calD_{\mathrm{cal}}\setminus\{q\}$, $V_i$ replaced by $S_i$, and $\hmu_t$ replaced by $\htau$; denote the resulting interval by $\hC^{\Gamma}(x_q)=[\htau(x_q)\pm\hQ^\Gamma(x_q)]$ and the localized reference measure (now relative to $P_X$) by $P_h$.
	
	The quantity $\hG_q$ is a noisy, observed-data transformation whose conditional mean approximates $\tau(x_q)$. Coverage of $\tau_q$ requires an additional assumption, which we state in the form introduced by \citet{alaa2023}.
	
	Let $H_x(s)=\Prob(|\hG-\htau(x)|\le s\mid X=x)$ be the conditional distribution function of the pseudo-outcome score, which is a functional of the observed-data distribution, and let $H^\ast_x(s)=\Prob(|Y(1)-Y(0)-\htau(x)|\le s\mid X=x)$ be that of the oracle ITE score, which depends on the joint law of $(Y(0),Y(1))$ given $X$.
	
	\begin{assumption}[Local stochastic dominance at $x$]
		\label{ass:dominance}
		$H^\ast_x(s)\ge H_x(s)$ for all $s\ge 0$; that is, the oracle ITE score at $x$ is stochastically smaller than the pseudo-outcome score at $x$.
	\end{assumption}
	
	\begin{theorem}[Coverage of the pseudo-outcome variant]
		\label{thm:pseudo}
		Suppose Assumptions~\ref{ass:sutva} and \ref{ass:ignorability} hold, and condition on $\calD_{\mathrm{tr}}$.
		\begin{enumerate}[(i)]
			\item If Assumption~\ref{ass:dominance} holds for $P_h$-almost every $x$, then for a new unit with covariates $X'\sim P_h$ and potential outcomes $(Y'(0),Y'(1))$, independent of the data,
			\[
			\Prob\bigl(Y'(1)-Y'(0)\in\hC^\Gamma(X')\bigr)\ge 1-\alpha .
			\]
			\item If Assumption~\ref{ass:dominance} holds at $x_q$, then for an out-of-sample query, or an in-sample query in $\calD_{\mathrm{cal}}$,
			\[
			\Prob\bigl(\tau_q\in\hC^\Gamma(x_q)\mid X_q=x_q\bigr)\ \ge\ 1-\alpha-\varepsilon^\Gamma(h;x_q),
			\]
			where $\varepsilon^\Gamma(h;x_q)=\int\sup_{s}|H_x(s)-H_{x_q}(s)|\,dP_h(x)$.
		\end{enumerate}
		Without Assumption~\ref{ass:dominance}, the same inequalities hold with $Y'(1)-Y'(0)$ and $\tau_q$ replaced by $\hG'$ and $\hG_q$.
	\end{theorem}
	
	\begin{remark}[On the plausibility and the status of Assumption~\ref{ass:dominance}]
		\label{rem:dominance}
		Assumption~\ref{ass:dominance} is a cross-world condition: $H^\ast_x$ depends on the dependence between $Y(0)$ and $Y(1)$ given $X$, which no observational design identifies, so it cannot be checked from data. If the nuisance functions are correct, so that $\htau=\tau$ and $\hmu_t=\mu_t$, and $\sigma_t^2(x)=\mathrm{Var}(Y(t)\mid X=x)$, then
		\begin{align*}
			\E\bigl[(\hG-\tau(X))^2\mid X=x\bigr]&=\frac{\sigma_1^2(x)}{e(x)}+\frac{\sigma_0^2(x)}{1-e(x)}\\
			&\ge\ \{\sigma_1(x)+\sigma_0(x)\}^2\ \ge\ \mathrm{Var}\bigl(Y(1)-Y(0)\mid X=x\bigr),
		\end{align*}
		where the first inequality is the Cauchy--Schwarz inequality applied to the decomposition $\sigma_1+\sigma_0=(\sigma_1/\sqrt e)\sqrt e+(\sigma_0/\sqrt{1-e})\sqrt{1-e}$, and the second is Minkowski's inequality, which holds whatever the dependence between $Y(0)$ and $Y(1)$. The pseudo-outcome score thus always has the larger conditional second moment, but second-moment dominance does not imply the stochastic dominance of Assumption~\ref{ass:dominance}. \citet{alaa2023} prove a distribution-free convex-order dominance for the doubly robust score and show that it suffices for validity at the coverage levels used in practice; their argument applies conditionally on $X=x$ and therefore to Assumption~\ref{ass:dominance}.
	\end{remark}
	
	\section{Implementation}
	\label{sec:impl}
	
	\subsection{Nuisance functions}
	
	On $\calD_{\mathrm{tr}}$ we fit a causal forest with the \texttt{grf} package \citep{grf2024}, which provides $\htau$, out-of-bag estimates of $\he$ and of the marginal outcome regression, variable-importance scores $\omega_1,\dots,\omega_p$, and the forest's adaptive neighborhood weights. Arm-specific outcome regressions $\hmu_t$ are fitted by regression forests on the arm-$t$ units of $\calD_{\mathrm{tr}}$.
	
	\subsection{Localization kernels}
	\label{sec:kernels}
	
	We consider two kernels, both fixed functions of $x$ given $\calD_{\mathrm{tr}}$.
	
	\begin{itemize}
		\item \emph{Importance-weighted Gaussian kernel.} With $\sigma_j$ the standard deviation of $X_j$ on $\calD_{\mathrm{tr}}$ and $\bar\omega_j=\omega_j/\sum_k\omega_k$,
		\begin{align*}
			d_\omega(x,x')^2&=\sum_{j=1}^p\bar\omega_j\Bigl(\frac{x_j-x'_j}{\sigma_j}\Bigr)^2,\\
			K_h(x,x_q)&=\exp\Bigl\{-\frac{d_\omega(x,x_q)^2}{2h^2}\Bigr\}\exp\Bigl\{-\frac{(\he(x)-\he(x_q))^2}{2h_e^2}\Bigr\}.
		\end{align*}
		The first factor concentrates the neighborhood on the covariates that drive effect heterogeneity, so that noise covariates do not dilute it. The second factor favors units with similar treatment probability. Both factors are bounded by one and equal one at $x=x_q$.
		\item \emph{Forest kernel.} $K(x,x_q)$ is the fraction of trees of the causal forest in which $x$ and $x_q$ fall in the same leaf. This is the adaptive neighborhood of \citet{athey2019grf} evaluated between a calibration point and the query. It adapts automatically to the directions along which $\tau$ varies and needs no bandwidth. Its effective bandwidth is controlled by the minimum leaf size.
	\end{itemize}
	
	\subsection{Choosing the bandwidth without touching calibration outcomes}
	\label{sec:bandwidth}
	
	The bandwidth $h$ governs the trade-off of Remark~\ref{rem:tradeoff}. It must be chosen as a function of $\calD_{\mathrm{tr}}$ and $x_q$ only; choosing it on the calibration fold, even using only covariates, makes the weights data dependent and voids Theorem~\ref{thm:local}. We use the following rule. Split $\calD_{\mathrm{tr}}$ once more into a fitting part (60\%) and a validation part (40\%), and fit arm-specific regressions and a propensity model on the fitting part. For each $h$ in a grid, and separately for each arm $t$, use every validation unit (of either arm) as a query point and the remaining arm-$t$ validation units as its calibration set (leave one out). At each query point record the kernel mass $W_t$ and the effective sample size $n_{\mathrm{eff},t}$, both rescaled by the ratio of the number of arm-$t$ units in $\calD_{\mathrm{cal}}$ to that in the validation part, since kernel mass grows linearly in the number of units; at query points of arm $t$, whose $Y(t)$ is observed, also record whether it falls in the interval. We select the smallest $h$ such that, at $95\%$ of the validation points and in both arms, the rescaled kernel mass is at least $(1-\alpha/2)/(\alpha/2)$, so that by Remark~\ref{rem:tradeoff} a finite interval exists even at the Bonferroni level used for out-of-sample queries, and the rescaled effective sample size is at least $n_{\min}=30$, while the validation coverage is at least $1-\alpha$ minus two standard errors. Both size conditions are needed: since $K\le1$, $n_{\mathrm{eff},t}\ge W_t$, so a bound on $n_{\mathrm{eff},t}$ alone does not prevent infinite intervals. In practice the coverage condition is rarely binding, because weighted conformal intervals are valid for every $h$; the size conditions determine $h$. The propensity bandwidth is fixed at $h_e=0.1$. A query-specific alternative, which requires no validation fit, chooses for each query the smallest $h$ for which the rescaled kernel mass of the arm-$t$ training units, $\sum_{i\in\calD_{\mathrm{tr}},T_i=t}K_h(X_i,x_q)$, reaches a target value in both arms; it is also a function of $\calD_{\mathrm{tr}}$ and $x_q$ only and is implemented in the accompanying code.
	
	\subsection{Algorithm}
	
	For reference, the complete procedure for a query $x_q$ at level $\alpha$ is:
	\begin{enumerate}[1.]
		\item Split $\calD$ into $\calD_{\mathrm{tr}}$ and $\calD_{\mathrm{cal}}$.
		\item On $\calD_{\mathrm{tr}}$: fit $\hmu_0,\hmu_1,\he,\htau$ and the variable-importance scores; build $K$; select $h$ as in Section~\ref{sec:bandwidth}.
		\item Counterfactual arm $t$: on $\calI_t\subset\calD_{\mathrm{cal}}$ compute $V_i$, $w_i$, the weighted quantile $\hQ_t(x_q)$, and $\hC_t(x_q)$.
		\item In-sample query: return $\hC(\tau_q)$ from Corollary~\ref{cor:ite}(a). Out-of-sample query: repeat step~3 for both arms at levels $\alpha/2$ and return the Bonferroni interval of Corollary~\ref{cor:ite}(b).
		\item Optionally: compute the pseudo-outcome interval $\hC^\Gamma(x_q)$ on $\calD_{\mathrm{cal}}\setminus\{q\}$ and report it as a secondary analysis.
		\item Report $n_{\mathrm{eff},t}$ and the largest normalized weight next to each interval.
	\end{enumerate}
	To serve every unit of $\calD$ as an in-sample query, the roles of the folds can be swapped and the procedure repeated.
	
	\section{Simulation study}
	\label{sec:sim}
	
	\subsection{Design}
	
	The design has three goals: genuine individual-level heterogeneity conditional on $X$, so that ITE coverage and CATE coverage differ; regions of the covariate space where the residual distribution differs, so that local coverage is informative; and noise covariates, so that localization in the full covariate space is harmful and localization in the relevant subspace is beneficial.
	
	Covariates $X\sim N(0,I_{20})$. Treatment $T\mid X\sim\mathrm{Bernoulli}(e(X))$ with
	\[
	e(x)=\min\bigl\{0.95,\max\{0.05,\ \mathrm{expit}(0.5x_1+0.8x_2^2-0.5)\}\bigr\}.
	\]
	Potential outcomes are
	\[
	Y(0)=\mu_0(X)+\epsilon_0,\qquad Y(1)=\mu_0(X)+\tau(X)+\sigma_\tau(X)\,\xi+\epsilon_1,
	\]
	with $\mu_0(x)=2\sin(\pi x_1)+x_2^2+\exp(0.5x_3)$, $\tau(x)=5+2x_1-3x_2^2+2\cos(\pi x_3)$, and $\sigma_\tau(x)=1+2\,\ind\{x_1>0\}$. The term $\sigma_\tau(X)\xi$, with $\xi\sim N(0,1)$ (Design~A) or $\xi\sim t_3/\sqrt3$ (Design~B) independent of everything else, is unit-level effect heterogeneity that is not a function of $X$; the ITE is $\tau_i=\tau(X_i)+\sigma_\tau(X_i)\xi_i+\epsilon_{i1}-\epsilon_{i0}$. The noise pair $(\epsilon_0,\epsilon_1)$ is bivariate normal with unit variances and correlation $\rho\in\{0,0.5\}$; $\rho$ affects the ITE distribution but is not identifiable. Only $x_1,x_2,x_3$ enter any structural function; the other 17 covariates are noise. Sample sizes are $n\in\{1000,4000\}$, split equally between folds, and each configuration is replicated $R=200$ times. In each replication, $500$ query units are drawn independently from $P$. Each query is evaluated twice: as an out-of-sample query, with both potential outcomes withheld, and as an in-sample query, with its factual $(T_q,Y_q)$ revealed.
	
	\subsection{Methods and metrics}
	
	All methods use the same training/calibration split where applicable, $\alpha=0.1$, and causal forests with 2000 trees.
	\begin{itemize}
		\item \textsc{CF-CI}: the asymptotic causal-forest confidence interval for $\tau(x)$, included to quantify how far a CATE interval falls below nominal ITE coverage.
		\item \textsc{Global}: unweighted split conformal on each arm ($h=\infty$); it ignores the covariate shift between arms and is not guaranteed to be valid for the target.
		\item \textsc{LC-Nested}, \textsc{LC-In}: the weighted split conformalized quantile regression of \citet{lei2021} with estimated propensity weights, using the nested construction for out-of-sample queries and the counterfactual construction for in-sample queries, as implemented in the \texttt{cfcausal} package \citep{cfcausal} with quantile random forests, fitted on the full sample.
		\item \textsc{CM-DR}: the conformal doubly robust meta-learner of \citet{alaa2023}, i.e., the pseudo-outcome procedure of Section~\ref{sec:pseudo} with $h=\infty$.
		\item \textsc{RLCP-IW}: the randomly localized conformal prediction of \citet{hore2025} applied arm by arm, with the importance-weighted Gaussian kernel (no propensity factor) and the same selected bandwidth as \textsc{Loc-IW0}.
		\item \textsc{Loc-Raw}: the proposed procedure with an unweighted Gaussian kernel in all 20 coordinates, to isolate the effect of importance weighting.
		\item \textsc{Loc-IW0}, \textsc{Loc-IW}: the proposed procedure with the importance-weighted Gaussian kernel, without and with the propensity factor.
		\item \textsc{Loc-F}: the proposed procedure with the forest kernel, computed from 300 trees of the causal forest.
		\item \textsc{Loc-$\Gamma$}: the localized pseudo-outcome variant of Section~\ref{sec:pseudo} with the \textsc{Loc-IW} kernel and bandwidth.
	\end{itemize}
	
	For each method and replication we record, separately for in-sample and out-of-sample queries: marginal coverage of $\tau_q$; coverage within the subgroups $\{x_1>0\}$ and $\{x_1\le 0\}$, within propensity tertiles, and within deciles of $\tau(X)$; slab coverage within each of the fifty slabs of $P_X$-mass $0.1$ given by the deciles of $x_1$, $x_2$, $x_3$, $\tau(X)$, and $e(X)$; mean interval length over finite intervals and the share of infinite intervals; the average effective sample size and largest normalized weight; and the root mean squared error of the interval center with respect to $\tau(X_q)$ and to $\tau_q$. We report means over replications with Monte Carlo standard errors. Because each slab contains only about fifty queries per replication, the per-replication minimum over fifty slabs is biased downward by sampling noise. We therefore report \emph{worst-slab coverage} as the minimum over slabs of the slab coverage pooled across replications. Sensitivity analyses vary $h$ over the grid $\{0.3,0.5,0.75,1,1.5,2,\infty\}$ and, for Design~A with $n=4000$ and $R=100$, the training fraction over $\{0.3,0.5,0.7\}$. Implementation details and the code that reproduces all results are described in Sections~S2 and S5 of the Supplementary Material \citep{supp}.
	
	\subsection{Results}
	
	Tables~\ref{tab:simA} and \ref{tab:simA_in} report Design~A with $\rho=0$ and $n=4000$; the other configurations lead to the same conclusions and are summarized in Table~\ref{tab:stress} and reported in full in Section~S3 of the Supplementary Material \citep{supp}. Monte Carlo standard errors are at most $0.002$ for coverage, so all differences discussed below are many standard errors wide. Mean lengths are computed over finite intervals. The share of infinite intervals is reported next to them and must be read together with them.
	
	As a point of reference, the causal-forest confidence interval for $\tau(x)$, which targets the conditional mean, contains the individual effect $\tau_q$ for $22\%$ of the queries (Table~\ref{tab:simA}), which illustrates the difference between CATE and ITE coverage that motivates this work.
	
	\begin{table}[t]
		\footnotesize
		\centering
		\caption{Design A ($\xi$ Gaussian, $\rho=0$), $n=4000$, $\alpha=0.10$, out-of-sample queries. Means over 200 replications, Monte Carlo standard errors in parentheses.}\label{tab:simA}
		\tabcolsep=4pt
		\begin{adjustbox}{max width=\textwidth}\begin{tabular}{lccccccc}
				\toprule
				Method & Marg.\ cov. & Cov.\ $x_1\le0$ & Cov.\ $x_1>0$ & Worst slab & Mean length & Inf.\ share & $n_{\mathrm{eff}}$ \\
				\midrule
				\textsc{CF-CI}     & 0.224 (0.002) & 0.247 (0.002) & 0.200 (0.002) & 0.125 & \phantom{0}2.01 (0.01) & 0.000 & --- \\
				\textsc{Global}    & 0.981 (0.001) & 0.990 (0.000) & 0.972 (0.001) & 0.869 & 18.56 (0.04) & 0.000 & 924 \\
				\textsc{LC-Nested} & 0.873 (0.002) & 0.914 (0.002) & 0.831 (0.002) & 0.579 & 11.32 (0.04) & 0.000 & --- \\
				\textsc{CM-DR}     & 0.945 (0.001) & 0.964 (0.001) & 0.926 (0.001) & 0.623 & 15.71 (0.04) & 0.000 & 2000 \\
				\textsc{RLCP-IW}   & 0.987 (0.001) & 0.994 (0.000) & 0.981 (0.001) & 0.942 & 18.00 (0.03) & 0.043 & 237 \\
				\textsc{Loc-Raw}   & 0.982 (0.001) & 0.990 (0.000) & 0.974 (0.001) & 0.873 & 18.62 (0.04) & 0.000 & 769 \\
				\textsc{Loc-IW0}   & 0.992 (0.000) & 0.998 (0.000) & 0.986 (0.001) & 0.977 & 17.65 (0.03) & 0.035 & 390 \\
				\textsc{Loc-IW}    & 0.993 (0.000) & 0.998 (0.000) & 0.987 (0.001) & 0.980 & 18.16 (0.04) & 0.047 & 360 \\
				\textsc{Loc-F}     & 0.992 (0.000) & 0.994 (0.000) & 0.990 (0.000) & 0.939 & 21.24 (0.06) & 0.004 & 214 \\
				\textsc{Loc-$\Gamma$} & 0.955 (0.001) & 0.968 (0.001) & 0.941 (0.001) & 0.698 & 15.46 (0.04) & 0.000 & 845 \\
				\bottomrule
		\end{tabular}\end{adjustbox}
		\tabnote{Marg.\ cov.: marginal coverage of $\tau_q$. Cov.\ $x_1\le0$, Cov.\ $x_1>0$: coverage within the subgroups. Worst slab: minimum over the 50 decile slabs of $x_1,x_2,x_3,\tau(X),e(X)$ of the coverage pooled across replications. Mean length: mean length of the finite intervals. Inf.\ share: share of infinite intervals. $n_{\mathrm{eff}}$: average effective sample size of the calibration weights. Method abbreviations are defined in Section~\ref{sec:sim}.}
	\end{table}
	
	\begin{table}[t]
		\footnotesize
		\centering
		\caption{Same design as Table~\ref{tab:simA}, in-sample queries (factual outcome observed, counterfactual predicted).}\label{tab:simA_in}
		\tabcolsep=4pt
		\begin{adjustbox}{max width=\textwidth}\begin{tabular}{lccccccc}
				\toprule
				Method & Marg.\ cov. & Cov.\ $x_1\le0$ & Cov.\ $x_1>0$ & Worst slab & Mean length & Inf.\ share & $n_{\mathrm{eff}}$ \\
				\midrule
				\textsc{Global}    & 0.853 (0.001) & 0.894 (0.001) & 0.811 (0.002) & 0.353 & \phantom{0}7.24 (0.02) & 0.000 & 994 \\
				\textsc{LC-In}     & 0.886 (0.001) & 0.908 (0.002) & 0.865 (0.002) & 0.654 & \phantom{0}8.35 (0.02) & 0.000 & --- \\
				\textsc{RLCP-IW}   & 0.870 (0.001) & 0.896 (0.002) & 0.844 (0.002) & 0.562 & \phantom{0}7.11 (0.02) & 0.019 & 257 \\
				\textsc{Loc-Raw}   & 0.852 (0.001) & 0.893 (0.001) & 0.811 (0.002) & 0.353 & \phantom{0}7.20 (0.02) & 0.000 & 824 \\
				\textsc{Loc-IW0}   & 0.863 (0.001) & 0.885 (0.002) & 0.841 (0.002) & 0.537 & \phantom{0}6.87 (0.01) & 0.013 & 413 \\
				\textsc{Loc-IW}    & 0.887 (0.001) & 0.908 (0.001) & 0.866 (0.002) & 0.714 & \phantom{0}7.19 (0.02) & 0.007 & 401 \\
				\textsc{Loc-F}     & 0.899 (0.001) & 0.914 (0.001) & 0.884 (0.002) & 0.583 & \phantom{0}7.66 (0.02) & 0.000 & 247 \\
				\textsc{Loc-$\Gamma$} & 0.955 (0.001) & 0.968 (0.001) & 0.941 (0.001) & 0.698 & 15.46 (0.04) & 0.000 & 845 \\
				\bottomrule
		\end{tabular}\end{adjustbox}
		\tabnote{Column abbreviations as in Table~\ref{tab:simA}.}
	\end{table}
	
	\runin{Marginal validity hides local failures.} Out of sample, the methods that calibrate globally attain or approach nominal marginal coverage but undercover where effects are most heterogeneous. \textsc{CM-DR} has marginal coverage $0.945$, but $0.926$ in $\{x_1>0\}$ and $0.623$ in the worst slab. \textsc{LC-Nested}, whose guarantee for out-of-sample ITEs is only approximate, undercovers marginally ($0.873$), in $\{x_1>0\}$ ($0.831$), and in the worst slab ($0.579$); it is the shortest method because it is the most anticonservative. The unweighted \textsc{Global} procedure overcovers marginally ($0.981$) because the Bonferroni combination of Corollary~\ref{cor:ite}(b) is conservative, but its worst-slab coverage is still $0.869$. In sample (Table~\ref{tab:simA_in}), where a single counterfactual interval is calibrated on the opposite arm, \textsc{Global} is not valid even marginally ($0.853$, and $0.353$ in the worst slab): the covariate distributions of the two arms differ, and ignoring that shift is costly. \textsc{LC-In}, which corrects the shift with estimated propensity weights, raises marginal coverage to $0.886$ and worst-slab coverage to $0.654$.
	
	\runin{Localization raises local coverage.} Out of sample, \textsc{Loc-IW} raises worst-slab coverage from $0.869$ (\textsc{Global}) and $0.579$ (\textsc{LC-Nested}) to $0.980$, with coverage $0.987$ in $\{x_1>0\}$ and a mean finite length ($18.2$) slightly below that of \textsc{Global} ($18.6$); the price is that $4.7\%$ of the intervals are infinite, at queries whose neighborhood contains too few calibration units of one arm (Remark~\ref{rem:tradeoff}). In sample, \textsc{Loc-IW} doubles the worst-slab coverage of \textsc{Global} ($0.714$ versus $0.353$), exceeds that of \textsc{LC-In} ($0.654$) with shorter intervals ($7.2$ versus $8.3$), and attains marginal coverage $0.887$ without any propensity weighting. It does not reach the nominal level: the remaining shortfall is the localization error $\varepsilon_t(h;x_q)$ of Theorem~\ref{thm:local}(ii), which is not negligible here because the selected bandwidth ($h=0.75$ in $58\%$ of the replications and $h=1$ in most of the others) is wide relative to the discontinuity of $\sigma_\tau$ at $x_1=0$. The forest kernel behaves differently: \textsc{Loc-F} has the highest in-sample marginal coverage among the counterfactual methods ($0.899$) but a lower worst-slab coverage ($0.583$), because the causal forest partitions the covariate space along the directions in which $\tau(x)$ varies, not those in which the residual dispersion varies.
	
	\runin{What drives the gains.} Without importance weights, \textsc{Loc-Raw} coincides with \textsc{Global} in every metric, although its selected bandwidth is the same as that of \textsc{Loc-IW}: in 20 dimensions, distances between points concentrate, so an unweighted Gaussian kernel assigns nearly the same weight to every calibration unit and localizes in no direction. The importance weights are what keeps the 17 noise covariates from diluting the neighborhood. The propensity factor matters in sample (\textsc{Loc-IW} versus \textsc{Loc-IW0}: worst slab $0.714$ versus $0.537$), where calibration and query arms differ, and is immaterial out of sample. Randomizing the kernel center (\textsc{RLCP-IW}) with the same kernel and bandwidth as \textsc{Loc-IW0} gives slightly lower local coverage out of sample ($0.942$ versus $0.977$ in the worst slab) and similar results in sample.
	
	\runin{The pseudo-outcome variant.} \textsc{Loc-$\Gamma$} is the shortest valid out-of-sample procedure ($15.5$, against $18.2$ for the Bonferroni interval of \textsc{Loc-IW}) and improves on its global counterpart \textsc{CM-DR} in every coverage metric. Table~\ref{tab:stress} shows that its marginal coverage stays between $0.945$ and $0.961$ under heavy-tailed effect heterogeneity (Design~B) and correlated potential-outcome noise ($\rho=0.5$), the two features that stress Assumption~\ref{ass:dominance}, so the assumption appears to hold in these designs. Its worst-slab coverage, however, stays near $0.70$. A plausible explanation is that the dispersion of the pseudo-outcome score is driven largely by the inverse-propensity factor, which differs between the arms pooled in its calibration set, so that localization in covariates corrects local miscalibration only partially.
	
	\begin{table}[t]
		\footnotesize
		\centering
		\caption{Stress test of Assumption~\ref{ass:dominance}: Designs A and B and $\rho\in\{0,0.5\}$, $n=4000$, out-of-sample queries.}\label{tab:stress}
		\tabcolsep=6pt
		\begin{adjustbox}{max width=\textwidth}\begin{tabular}{llcccc}
				\toprule
				Setting & Method & Marg.\ cov. & Cov.\ $x_1>0$ & Worst slab & Mean length \\
				\midrule
				A, $\rho=0.0$ & \textsc{Global}    & 0.981 (0.001) & 0.972 (0.001) & 0.869 & 18.56 (0.04) \\
				A, $\rho=0.0$ & \textsc{LC-Nested} & 0.873 (0.002) & 0.831 (0.002) & 0.579 & 11.32 (0.04) \\
				A, $\rho=0.0$ & \textsc{CM-DR}     & 0.945 (0.001) & 0.926 (0.001) & 0.623 & 15.71 (0.04) \\
				A, $\rho=0.0$ & \textsc{Loc-IW}    & 0.993 (0.000) & 0.987 (0.001) & 0.980 & 18.16 (0.04) \\
				A, $\rho=0.0$ & \textsc{Loc-F}     & 0.992 (0.000) & 0.990 (0.000) & 0.939 & 21.24 (0.06) \\
				A, $\rho=0.0$ & \textsc{Loc-$\Gamma$} & 0.955 (0.001) & 0.941 (0.001) & 0.698 & 15.46 (0.04) \\
				A, $\rho=0.5$ & \textsc{Global}    & 0.983 (0.000) & 0.977 (0.001) & 0.880 & 18.47 (0.04) \\
				A, $\rho=0.5$ & \textsc{LC-Nested} & 0.884 (0.001) & 0.840 (0.002) & 0.606 & 11.24 (0.04) \\
				A, $\rho=0.5$ & \textsc{CM-DR}     & 0.951 (0.001) & 0.934 (0.001) & 0.644 & 15.75 (0.04) \\
				A, $\rho=0.5$ & \textsc{Loc-IW}    & 0.995 (0.000) & 0.990 (0.000) & 0.985 & 18.11 (0.04) \\
				A, $\rho=0.5$ & \textsc{Loc-F}     & 0.994 (0.000) & 0.992 (0.000) & 0.947 & 21.17 (0.06) \\
				A, $\rho=0.5$ & \textsc{Loc-$\Gamma$} & 0.961 (0.001) & 0.948 (0.001) & 0.718 & 15.49 (0.04) \\
				B, $\rho=0.0$ & \textsc{Global}    & 0.975 (0.001) & 0.964 (0.001) & 0.858 & 17.57 (0.04) \\
				B, $\rho=0.0$ & \textsc{LC-Nested} & 0.876 (0.001) & 0.843 (0.002) & 0.566 & 10.64 (0.04) \\
				B, $\rho=0.0$ & \textsc{CM-DR}     & 0.935 (0.001) & 0.913 (0.002) & 0.582 & 14.26 (0.04) \\
				B, $\rho=0.0$ & \textsc{Loc-IW}    & 0.988 (0.000) & 0.979 (0.001) & 0.974 & 17.18 (0.04) \\
				B, $\rho=0.0$ & \textsc{Loc-F}     & 0.989 (0.000) & 0.985 (0.001) & 0.944 & 20.99 (0.07) \\
				B, $\rho=0.0$ & \textsc{Loc-$\Gamma$} & 0.945 (0.001) & 0.927 (0.001) & 0.684 & 14.03 (0.04) \\
				B, $\rho=0.5$ & \textsc{Global}    & 0.978 (0.001) & 0.968 (0.001) & 0.866 & 17.57 (0.04) \\
				B, $\rho=0.5$ & \textsc{LC-Nested} & 0.893 (0.002) & 0.864 (0.002) & 0.579 & 10.70 (0.04) \\
				B, $\rho=0.5$ & \textsc{CM-DR}     & 0.939 (0.001) & 0.919 (0.002) & 0.590 & 14.19 (0.04) \\
				B, $\rho=0.5$ & \textsc{Loc-IW}    & 0.990 (0.000) & 0.982 (0.001) & 0.977 & 17.21 (0.04) \\
				B, $\rho=0.5$ & \textsc{Loc-F}     & 0.990 (0.000) & 0.986 (0.001) & 0.944 & 20.91 (0.07) \\
				B, $\rho=0.5$ & \textsc{Loc-$\Gamma$} & 0.950 (0.001) & 0.934 (0.001) & 0.690 & 13.97 (0.04) \\
				\bottomrule
		\end{tabular}\end{adjustbox}
		\tabnote{Design A: $\xi\sim N(0,1)$; Design B: $\xi\sim t_3/\sqrt3$; $\rho$: correlation of the potential-outcome noise. Column abbreviations as in Table~\ref{tab:simA}.}
	\end{table}
	
	\begin{table}[t]
		\footnotesize
		\centering
		\caption{Sensitivity of \textsc{Loc-IW} to the bandwidth (Design A, $n=4000$, out-of-sample queries).}\label{tab:sens}
		\tabcolsep=8pt
		\begin{adjustbox}{max width=\textwidth}\begin{tabular}{cccccc}
				\toprule
				$h$ & Marg.\ cov. & Worst slab & Mean length & $n_{\mathrm{eff}}$ & Inf.\ share \\
				\midrule
				0.30 & 0.998 (0.000) & 0.993 & 19.16 (0.06) & 105 & 0.478 \\
				0.50 & 0.995 (0.000) & 0.985 & 17.62 (0.04) & 225 & 0.181 \\
				0.75$^{a}$ & 0.994 (0.000) & 0.981 & 18.01 (0.04) & 331 & 0.072 \\
				1.00 & 0.991 (0.000) & 0.976 & 18.27 (0.04) & 394 & 0.024 \\
				1.50 & 0.987 (0.000) & 0.938 & 18.38 (0.04) & 450 & 0.003 \\
				2.00 & 0.986 (0.000) & 0.927 & 18.40 (0.04) & 471 & 0.000 \\
				$\infty$ & 0.981 (0.001) & 0.869 & 18.56 (0.04) & 924 & 0.000 \\
				\bottomrule
		\end{tabular}\end{adjustbox}
		\tabnote{$^{a}$Bandwidth most often selected by the training-fold rule. $h=\infty$ is the global procedure. Column abbreviations as in Table~\ref{tab:simA}.}
	\end{table}
	
	\runin{Tuning.} Table~\ref{tab:sens} displays the trade-off of Remark~\ref{rem:tradeoff}. As $h$ decreases from $\infty$ to $0.3$, worst-slab coverage rises from $0.869$ to $0.993$ and the share of infinite intervals rises from zero to $48\%$; the mean finite length is minimized at an intermediate bandwidth. The training-fold rule selects $h=0.75$, near the point where worst-slab coverage stops improving appreciably and infinite intervals are still rare. Table~\ref{tab:split} shows that the conclusions do not depend on the training fraction, with one exception: with $70\%$ of the sample in the training fold, the forest kernel becomes so concentrated relative to the smaller calibration fold that $99.7\%$ of the out-of-sample intervals of \textsc{Loc-F} are infinite. Unlike the Gaussian kernels, the forest kernel has no bandwidth that the selection rule can adjust, and its leaf size must be tuned to the size of the calibration fold.
	
	\begin{table}[t]
		\footnotesize
		\centering
		\caption{Sensitivity to the training fraction (Design A, $\rho=0$, $n=4000$, 100 replications).}\label{tab:split}
		\tabcolsep=3.5pt
		\begin{adjustbox}{max width=\textwidth}\begin{tabular}{llcccccc}
				\toprule
				Queries & Method & Training fraction & Marg.\ cov. & Cov.\ $x_1>0$ & Worst slab & Mean length & Inf.\ share \\
				\midrule
				ins & \textsc{Global}    & 0.3 & 0.854 (0.002) & 0.808 (0.003) & 0.347 & \phantom{0}7.65 (0.02) & 0.000 \\
				ins & \textsc{Global}    & 0.5 & 0.854 (0.002) & 0.813 (0.003) & 0.354 & \phantom{0}7.26 (0.02) & 0.000 \\
				ins & \textsc{Global}    & 0.7 & 0.854 (0.002) & 0.810 (0.003) & 0.352 & \phantom{0}7.05 (0.03) & 0.000 \\
				ins & \textsc{LC-In}     & 0.3 & 0.887 (0.002) & 0.862 (0.003) & 0.651 & \phantom{0}8.38 (0.04) & 0.000 \\
				ins & \textsc{LC-In}     & 0.5 & 0.886 (0.002) & 0.865 (0.003) & 0.651 & \phantom{0}8.34 (0.04) & 0.000 \\
				ins & \textsc{LC-In}     & 0.7 & 0.884 (0.002) & 0.862 (0.003) & 0.650 & \phantom{0}8.30 (0.03) & 0.000 \\
				ins & \textsc{Loc-IW}    & 0.3 & 0.883 (0.002) & 0.855 (0.003) & 0.689 & \phantom{0}7.47 (0.02) & 0.007 \\
				ins & \textsc{Loc-IW}    & 0.5 & 0.888 (0.002) & 0.867 (0.002) & 0.706 & \phantom{0}7.18 (0.02) & 0.007 \\
				ins & \textsc{Loc-IW}    & 0.7 & 0.889 (0.002) & 0.863 (0.003) & 0.691 & \phantom{0}7.10 (0.03) & 0.006 \\
				ins & \textsc{Loc-F}     & 0.3 & 0.879 (0.002) & 0.854 (0.002) & 0.497 & \phantom{0}7.66 (0.02) & 0.000 \\
				ins & \textsc{Loc-F}     & 0.5 & 0.900 (0.002) & 0.884 (0.002) & 0.578 & \phantom{0}7.66 (0.03) & 0.000 \\
				ins & \textsc{Loc-F}     & 0.7 & 0.944 (0.002) & 0.936 (0.002) & 0.766 & \phantom{0}8.90 (0.05) & 0.025 \\
				ins & \textsc{Loc-$\Gamma$} & 0.3 & 0.955 (0.001) & 0.943 (0.002) & 0.701 & 15.92 (0.07) & 0.000 \\
				ins & \textsc{Loc-$\Gamma$} & 0.5 & 0.956 (0.001) & 0.943 (0.002) & 0.704 & 15.48 (0.06) & 0.000 \\
				ins & \textsc{Loc-$\Gamma$} & 0.7 & 0.953 (0.001) & 0.936 (0.002) & 0.703 & 15.07 (0.07) & 0.000 \\
				oos & \textsc{Global}    & 0.3 & 0.980 (0.001) & 0.972 (0.001) & 0.852 & 19.47 (0.06) & 0.000 \\
				oos & \textsc{Global}    & 0.5 & 0.981 (0.001) & 0.972 (0.001) & 0.870 & 18.61 (0.06) & 0.000 \\
				oos & \textsc{Global}    & 0.7 & 0.981 (0.001) & 0.971 (0.001) & 0.881 & 18.06 (0.07) & 0.000 \\
				oos & \textsc{LC-Nested} & 0.3 & 0.872 (0.002) & 0.829 (0.003) & 0.577 & 11.30 (0.06) & 0.000 \\
				oos & \textsc{LC-Nested} & 0.5 & 0.873 (0.002) & 0.832 (0.003) & 0.583 & 11.34 (0.06) & 0.000 \\
				oos & \textsc{LC-Nested} & 0.7 & 0.870 (0.002) & 0.826 (0.003) & 0.577 & 11.27 (0.06) & 0.000 \\
				oos & \textsc{Loc-IW}    & 0.3 & 0.993 (0.000) & 0.988 (0.001) & 0.977 & 18.75 (0.05) & 0.031 \\
				oos & \textsc{Loc-IW}    & 0.5 & 0.993 (0.000) & 0.987 (0.001) & 0.980 & 18.17 (0.06) & 0.048 \\
				oos & \textsc{Loc-IW}    & 0.7 & 0.992 (0.000) & 0.986 (0.001) & 0.979 & 17.86 (0.07) & 0.059 \\
				oos & \textsc{Loc-F}     & 0.3 & 0.988 (0.001) & 0.984 (0.001) & 0.905 & 20.10 (0.06) & 0.000 \\
				oos & \textsc{Loc-F}     & 0.5 & 0.992 (0.001) & 0.990 (0.001) & 0.939 & 21.24 (0.08) & 0.004 \\
				oos & \textsc{Loc-F}     & 0.7 & 1.000 (0.000) & 1.000 (0.000) & 1.000 & 28.68 (1.36) & 0.997 \\
				oos & \textsc{Loc-$\Gamma$} & 0.3 & 0.955 (0.001) & 0.943 (0.002) & 0.701 & 15.92 (0.07) & 0.000 \\
				oos & \textsc{Loc-$\Gamma$} & 0.5 & 0.956 (0.001) & 0.943 (0.002) & 0.704 & 15.48 (0.06) & 0.000 \\
				oos & \textsc{Loc-$\Gamma$} & 0.7 & 0.953 (0.001) & 0.936 (0.002) & 0.703 & 15.07 (0.07) & 0.000 \\
				\bottomrule
		\end{tabular}\end{adjustbox}
		\tabnote{ins: in-sample queries; oos: out-of-sample queries. Training fraction: share of the sample in the training fold. Column abbreviations as in Table~\ref{tab:simA}.}
	\end{table}
	
	The same patterns hold at $n=1000$, with smaller gains: in sample, \textsc{Loc-IW} raises worst-slab coverage from $0.433$ (\textsc{Global}) and $0.622$ (\textsc{LC-In}) to $0.656$, and the selected bandwidth is larger ($h=1$), as the existence condition requires.
	
	\section{IHDP benchmark}
	\label{sec:ihdp}
	
	The Infant Health and Development Program (IHDP) data \citep{hill2011} combine covariates from a randomized study ($n=747$, $p=25$) with simulated outcomes. We use setting~A of the NPCI simulation, which corresponds to response surface~B of \citet{hill2011}, in the version distributed by \citet{shalit2017}, and analyze its first 50 replications. Each replication provides both potential outcomes, so the true ITE $\tau_i=Y_i(1)-Y_i(0)$, including outcome noise, is known for every unit. Within a replication we use two-fold cross-fitting: each fold in turn is the calibration fold, its units are the queries, and the other fold is used for all fitting and for bandwidth selection, so every unit receives one interval. Queries are treated as out-of-sample (both outcomes withheld), and a query's own score is removed from the calibration set. \textsc{LC-Nested} is fitted on the training fold only, so that it never sees the query. Local coverage is evaluated within tertiles of the estimated propensity score and within the two levels of the covariate with the largest causal-forest importance. With only about 70 treated units per calibration fold, the existence condition of Remark~\ref{rem:tradeoff} binds for narrow bandwidths, and the share of infinite intervals is reported.
	
	\begin{table}[t]
		\footnotesize
		\centering
		\caption{IHDP (setting A of NPCI, i.e., response surface B), 50 replications, $\alpha=0.10$. Queries are all units, each treated as out-of-sample with both outcomes withheld, using two-fold cross-fitting. Means (standard deviations) over replications.}\label{tab:ihdp}
		\tabcolsep=3.5pt
		\begin{adjustbox}{max width=\textwidth}\begin{tabular}{lcccccc}
				\toprule
				Method & Marg.\ cov. & Worst prop.\ tertile & Worst level top cov. & Mean length & Inf.\ share & RMSE center \\
				\midrule
				\textsc{CF-CI}     & 0.212 (0.034) & 0.174 (0.035) & 0.188 (0.035) & \phantom{0}2.904 (4.086) & 0.000 (0.000) & 5.077 (8.149) \\
				\textsc{Global}    & 0.984 (0.013) & 0.971 (0.023) & 0.970 (0.025) & 20.903 (23.172) & 0.000 (0.000) & 3.696 (6.313) \\
				\textsc{LC-Nested} & 0.813 (0.077) & 0.774 (0.090) & 0.781 (0.077) & \phantom{0}9.623 (9.256) & 0.000 (0.000) & 4.525 (7.295) \\
				\textsc{CM-DR}     & 0.939 (0.044) & 0.913 (0.060) & 0.894 (0.086) & 16.226 (16.558) & 0.000 (0.000) & 5.077 (8.149) \\
				\textsc{RLCP-IW}   & 0.987 (0.012) & 0.977 (0.018) & 0.976 (0.021) & 21.224 (23.872) & 0.071 (0.031) & 3.696 (6.313) \\
				\textsc{Loc-IW}    & 0.988 (0.010) & 0.978 (0.017) & 0.978 (0.018) & 21.104 (23.288) & 0.009 (0.018) & 3.696 (6.313) \\
				\textsc{Loc-F}     & 0.990 (0.009) & 0.981 (0.016) & 0.983 (0.015) & 21.548 (22.996) & 0.017 (0.051) & 3.696 (6.313) \\
				\textsc{Loc-$\Gamma$} & 0.937 (0.051) & 0.917 (0.062) & 0.902 (0.081) & 16.125 (16.299) & 0.000 (0.000) & 5.077 (8.149) \\
				\bottomrule
		\end{tabular}\end{adjustbox}
		\tabnote{Worst prop.\ tertile: lowest coverage across tertiles of the estimated propensity score. Worst level top cov.: lowest coverage across the two levels of the covariate with the largest causal-forest importance. RMSE center: root mean squared error of the interval center with respect to $\mu_1-\mu_0$. Other abbreviations as in Table~\ref{tab:simA}.}
	\end{table}
	
	Table~\ref{tab:ihdp} confirms the main findings at a much smaller sample size. The causal-forest confidence interval covers the ITE for $21\%$ of the units. \textsc{LC-Nested} undercovers (marginal coverage $0.813$, and $0.774$ in the worst propensity tertile), while the Bonferroni-combined conformal intervals overcover ($0.984$ to $0.990$). \textsc{CM-DR} and \textsc{Loc-$\Gamma$} are close to nominal ($0.939$ and $0.937$) and about $20\%$ shorter. With about 70 treated units per calibration fold, the selection rule chooses wide bandwidths, and \textsc{Loc-IW} is close to \textsc{Global}, as it should be when there is not enough local information to exploit. Interval lengths vary across replications by an order of magnitude because the scale of the simulated outcomes does, which explains the large standard deviations of the lengths.
	
	\section{Discussion}
	\label{sec:discussion}
	
	The question that motivates this paper is what can be said about the treatment effect of a particular unit rather than of a conditional mean. Conformal calibration localized around $x_q$ delivers exact finite-sample coverage for units drawn from a neighborhood of $x_q$, and coverage at $x_q$ itself at least $1-\alpha$ minus a localization error that is small when the distribution of the conformity score is smooth in the covariates and the neighborhood is narrow. The neighborhood cannot be arbitrarily narrow, because the weighted quantile needs enough effective calibration mass to be finite and stable.
	
	The simulations point to three practical lessons. First, the localization error is visible at realistic sample sizes. Second, the Bonferroni construction for out-of-sample queries is valid but conservative. Third, the importance-weighted Gaussian kernel was more reliable than the forest kernel. For applied work, we recommend reporting the effective sample size and the largest normalized weight next to each interval, so that a reader can see how much local information an interval rests on, and treating an infinite interval as an honest statement that the data contain too few comparable units of one arm.
	
	The main limitations are those of the identifying assumptions and of the localization error. Under unmeasured confounding, $P_{Y(t)\mid X,T=t}$ differs from $P_{Y(t)\mid X}$. The sensitivity frameworks of \citet{jin2023} and \citet{yin2024} can be localized along the lines of Remark~\ref{rem:propensity}. In high dimensions, a kernel that does not weight covariates by relevance fails to localize at all, as \textsc{Loc-Raw} illustrates, and importance weights estimated from a CATE learner localize along the directions in which the conditional mean of the effect varies. Extensions to instrumental-variable identification, to continuous treatments, and to longitudinal data require new identification arguments but not new calibration arguments, since the localized weighted quantile applies to any score that is exchangeable within arms.
	
	\begin{appendix}
		\section*{Proofs}\label{app:proofs}
		
		Throughout, we condition on $\calD_{\mathrm{tr}}$ and on the treatment indicators of $\calD_{\mathrm{cal}}$, so $\hmu_t$, $\htau$, $\he$, $K$, and $n_t$ are fixed.
		
		\begin{proof}[Proof of Theorem~\ref{thm:local}]
			\emph{Part (i).} By Assumptions~\ref{ass:sutva} and \ref{ass:ignorability}, for $i\in\calI_t$ the pair $(X_i,Y_i)=(X_i,Y_i(t))$ is a draw from $P_{X\mid t}\otimes P_{Y(t)\mid X}$, and these pairs are i.i.d. The test pair $(X',Y'(t))$ is drawn from $P_h^{(t)}\otimes P_{Y(t)\mid X}$, independently of the calibration pairs. The two laws share the conditional law of the outcome and differ by the covariate shift $dP_h^{(t)}/dP_{X\mid t}(x)=K(x,x_q)/c_t$ with $c_t=\int K(x,x_q)\,dP_{X\mid t}(x)\in(0,\infty)$. Corollary~1 of \citet{tibshirani2019} states that, under such a covariate shift with likelihood ratio $w$, the weighted conformal set with weights $p_i(x)=w(X_i)/\{\sum_j w(X_j)+w(x)\}$ and $p_\infty(x)=w(x)/\{\sum_j w(X_j)+w(x)\}$ has coverage at least $1-\alpha$. The normalizing constant $c_t$ cancels in these ratios, so the weights coincide with those of Section~\ref{sec:procedure}, and the lower bound follows. For the upper bound, the same weighted exchangeability argument shows that the event $\{V'\le\hQ_t(X')\}$ coincides with the event that $V'$ lies at or below the $(1-\alpha)$-quantile of the weighted empirical distribution of $(V_i)_{i\in\calI_t}$ and $V'$ with weights $p_i(X')$ and $p_\infty(X')$; when the scores are almost surely distinct this probability is at most $1-\alpha$ plus the weight of the atom at the quantile, which is at most $\max_i p_i(X')\vee p_\infty(X')$.
			
			\emph{Part (ii).} Write $\bar Q_t=\hQ_t(x_q)$ and $\bar C_t(x)=[\hmu_t(x)\pm\bar Q_t]$. For every $x$, $K(x,x_q)\le\bar w$ implies $p_i(x)\ge p_i(x_q)$ for all $i\in\calI_t$, hence $\sum_i p_i(x)\ind\{V_i\le s\}\ge\sum_i p_i(x_q)\ind\{V_i\le s\}$ for every $s$, hence $\hQ_t(x)\le\bar Q_t$ and $\hC_t(x)\subseteq\bar C_t(x)$. By part (i), $\Prob(Y'(t)\in\bar C_t(X'))\ge 1-\alpha$. Since $(X',Y'(t))$ is independent of the calibration pairs and $\bar Q_t$ is a function of them,
			\[
			\Prob\bigl(Y'(t)\in\bar C_t(X')\bigr)=\E\Bigl[\int G_{t,x}(\bar Q_t)\,dP_h^{(t)}(x)\Bigr].
			\]
			For the query, consider first the out-of-sample case. Then $Y_q(t)$, given $X_q=x_q$, has law $P_{Y(t)\mid X=x_q}$ and is independent of the calibration pairs, so $\Prob(Y_q(t)\in\hC_t(x_q)\mid X_q=x_q)=\E[G_{t,x_q}(\bar Q_t)]$. In the in-sample case $q\in\calD_{\mathrm{cal}}$ with $T_q=1-t$, unit $q$ is not in $\calI_t$, the calibration pairs are independent of $(X_q,T_q,Y_q(t))$, and by Assumption~\ref{ass:ignorability} the law of $Y_q(t)$ given $X_q=x_q$ and $T_q=1-t$ is again $P_{Y(t)\mid X=x_q}$; the same identity holds. Therefore
			\begin{align*}
				\Prob\bigl(Y_q(t)\in\hC_t(x_q)\mid X_q=x_q\bigr)&\ \ge\ \Prob\bigl(Y'(t)\in\bar C_t(X')\bigr)\\
				&\qquad-\E\Bigl[\int\bigl|G_{t,x}(\bar Q_t)-G_{t,x_q}(\bar Q_t)\bigr|\,dP_h^{(t)}(x)\Bigr]\\
				&\ \ge\ 1-\alpha-\varepsilon_t(h;x_q),
			\end{align*}
			where the last step bounds the integrand by its supremum over $s$.
		\end{proof}
		
		\begin{proof}[Proof of Corollary~\ref{cor:ite}]
			(a) If $T_q=0$ then $Y_q=Y_q(0)$ is observed, so the events $\{\tau_q\in\hC_1(x_q)-Y_q\}$ and $\{Y_q(1)\in\hC_1(x_q)\}$ coincide; apply Theorem~\ref{thm:local}(ii) with $t=1$. The case $T_q=1$ is symmetric. (b) Since\[\{\tau_q\notin[L_1-U_0,U_1-L_0]\}\subseteq\{Y_q(1)\notin\hC_1(x_q)\}\cup\{Y_q(0)\notin\hC_0(x_q)\},\] the union bound with Theorem~\ref{thm:local}(ii) for each arm gives the claim.
		\end{proof}
		
		\begin{proof}[Proof of Proposition~\ref{prop:lipschitz}]
			For fixed $x$ and $s$, we have $G_{t,x}(s)=L_{t,x}([-s,s])$ and $G_{t,x_q}(s)=L_{t,x_q}([-s,s])$, so that\[|G_{t,x}(s)-G_{t,x_q}(s)|\le\sup_A|L_{t,x}(A)-L_{t,x_q}(A)|=\TV(L_{t,x},L_{t,x_q}).\] Integrate against $P_h^{(t)}$ and apply the Lipschitz bound.
		\end{proof}
		
		\begin{proof}[Proof of Theorem~\ref{thm:pseudo}]
			The pairs $(X_i,\hG_i)$, $i\in\calD_{\mathrm{cal}}\setminus\{q\}$, are i.i.d.\ from $P_X\otimes P_{\hG\mid X}$, where $P_{\hG\mid X}$ is the conditional law of (\ref{eq:aipw}), a fixed transformation of $(X,T,Y)$ given $\calD_{\mathrm{tr}}$. A new unit with $X'\sim P_h$ and $(T',Y')\mid X'$ from the observed-data conditional law has $\hG'\mid X'\sim P_{\hG\mid X}$; the covariate shift has likelihood ratio proportional to $K(\cdot,x_q)$. Hence, exactly as in the proof of Theorem~\ref{thm:local}(i), $\Prob(|\hG'-\htau(X')|\le\hQ^\Gamma(X'))\ge 1-\alpha$, and with $\bar Q^\Gamma=\hQ^\Gamma(x_q)\ge\hQ^\Gamma(x)$,
			\[
			\E\Bigl[\int H_x(\bar Q^\Gamma)\,dP_h(x)\Bigr]\ \ge\ \E\Bigl[\int H_x(\hQ^\Gamma(x))\,dP_h(x)\Bigr]\ \ge\ 1-\alpha.
			\]
			\emph{Part (i).} Conditional on $X'=x$ and on the calibration data, $\hQ^\Gamma(x)$ is fixed and $\Prob(|Y'(1)-Y'(0)-\htau(x)|\le\hQ^\Gamma(x)\mid X'=x,\text{cal})=H^\ast_x(\hQ^\Gamma(x))\ge H_x(\hQ^\Gamma(x))$ by Assumption~\ref{ass:dominance} at $x$; the potential outcomes of the new unit are independent of the calibration data. Integrating over $x\sim P_h$ and over the calibration data gives $\Prob(Y'(1)-Y'(0)\in\hC^\Gamma(X'))\ge\E[\int H_x(\hQ^\Gamma(x))\,dP_h(x)]\ge 1-\alpha$.
			
			\emph{Part (ii).} For the query, $(Y_q(0),Y_q(1))$ given $X_q=x_q$ is independent of the calibration data (for an in-sample query because $q$ has been removed from the calibration set, and because $\hG_i$ for $i\neq q$ does not involve unit $q$), so
			\begin{align*}
				\Prob\bigl(\tau_q\in\hC^\Gamma(x_q)\mid X_q=x_q\bigr)&=\E\bigl[H^\ast_{x_q}(\bar Q^\Gamma)\bigr]\ \ge\ \E\bigl[H_{x_q}(\bar Q^\Gamma)\bigr]\\
				&\ge\ \E\Bigl[\int H_x(\bar Q^\Gamma)\,dP_h(x)\Bigr]-\varepsilon^\Gamma(h;x_q)\ \ge\ 1-\alpha-\varepsilon^\Gamma(h;x_q),
			\end{align*}
			using Assumption~\ref{ass:dominance} at $x_q$ in the first inequality, the definition of $\varepsilon^\Gamma$ in the second, and the display above in the third. The final sentence of the theorem follows by omitting the step that invokes Assumption~\ref{ass:dominance}.
		\end{proof}
	\end{appendix}
	
	\begin{acks}[Acknowledgments]
		Fernando Delbianco is the corresponding author. We thank two anonymous reviewers of an earlier version for comments that led to substantial changes in the theory and the experimental design.
	\end{acks}
	
	
	\begin{supplement}
		\stitle{Supplement to ``Localized conformal inference for individual treatment effects''}
		\sdescription{Proof of the location-scale example (Lemma S1), implementation details, results for all simulation configurations and additional metrics, additional details on the IHDP benchmark, and a description of the code (PDF).}
	\end{supplement}
	
	\begin{supplement}
		\stitle{Code and data}
		\sdescription{Zip archive with the R code that reproduces all simulations, the IHDP analysis, and every table of the paper and of the supplement; the 50 IHDP replications used; and the saved results of the full runs. A README file lists the commands.}
	\end{supplement}
	
	\section*{Statement on AI use}
	Anthropic's Fable 5.2 and Open AI's Astra have been used in different stages of the analysis. The full responsibility for the content corresponds to the authors.

\end{document}